\documentclass[
	aps,prd,
	onecolumn, 
	a4paper, 
	10pt, 
	amsmath,
	amssymb, 
	preprintnumbers,  
	tightenlines, 
	notitlepage,
	floatfix,
	nofootinbib,
	longbibliography
	]{revtex4-1}\pdfoutput=1
\usepackage[utf8]{inputenc}
\usepackage[english]{babel}
\usepackage{bbm}
\usepackage{graphicx}
\usepackage{hyperref}
\usepackage{wasysym}
\usepackage{amsthm}
\usepackage{color}	
\usepackage{tensor}

\graphicspath{{pics/}}

\allowdisplaybreaks[1]

\input{CQGformat.input}
\input{commands.input}

\begin{document}

\title{The mass and angular momentum of reconstructed metric perturbations}
\author{Maarten \surname{van de Meent}}
\address{Mathematical Sciences, University of Southampton, United Kingdom}
\ead{M.vandeMeent@soton.ac.uk}
\date{\today}
\begin{abstract}
We prove a key result regarding the mass and angular momentum content of linear vacuum perturbations of the Kerr metric obtained through the formalism developed by Chrzarnowski, Cohen, and Kegeles (CCK). More precisely we prove that the Abbott-Deser mass and angular momentum integrals of any such perturbation vanish, when that perturbation was obtained from a regular Fourier mode of the Hertz potential. As a corollary we obtain a generalization of previous results on the completion of the `no string' radiation gauge metric perturbation generated by a point particle. We find that for any bound orbit around a Kerr black hole, the mass and angular momentum perturbations completing the CCK metric are simply the energy and angular momentum of the particle ``outside'' the orbit and vanish ``inside'' the orbit.
\end{abstract}

\setlength{\parindent}{0pt} 
\setlength{\parskip}{6pt}

\section{Introduction}
In the 1970s, Chrzarnowski\cite{Chrzanowski:1975wv}, Cohen and Kegeles \cite{Cohen:1974cm,Kegeles:1979an}, and Wald \cite{Wald:1978vm} developed a procedure for obtaining vacuum solutions of the linearized Einstein equation on a Kerr background from a solution of the homogeneous spin $\pm2$ Teukolsky equation. The main advantages of this construction are that the Teukolsky equation is a scalar equation that (unlike the linearized Einstein equation in Kerr) can be separated over harmonic modes into a set of uncoupled ODEs.

As such, the procedure has found many applications in black hole perturbation theory~\cite{Yunes:2005ve,Dotti:2008yr,Amsel:2009ev,Dias:2009ex}. In recent years it has been shown that the CCK procedure --- as it is sometimes called --- can be used to reconstruct the metric perturbation caused by a point particle moving on a geodesic around a Kerr black hole from the Weyl scalars $\psi_0$ or $\psi_4$ sourced by the particle \cite{Ori:2002uv,Lousto:2002em,Keidl:2006wk,Keidl:2010pm,Shah:2010bi,Shah:2012gu,vandeMeent:2015lxa}. This perturbation can then be used to obtain the correction to equations of motion of the particle due to its own gravitational field in the form of the gravitational self force \cite{Pound:2013faa,vandeMeent:2016pee}. In turn, these self-force corrections play a crucial role in a faithful modelling of the dynamics of extreme mass-ratio inspirals (EMRIs), and their waveforms \cite{Hinderer:2008dm}.

A subtlety in reconstructing the metric perturbation from  $\psi_0$ or $\psi_4$ is that this procedure is (fundamentally) ambiguous up to metric perturbations for which  $\psi_0$ or $\psi_4$ vanish. Wald \cite{Wald:1973} showed that up to gauge modes this ambiguity formed solely of perturbations of the background metric within the family of Petrov type-D vacuum spacetimes. Further restricting to global vacuum solutions leaves only perturbations of $\delta M$ and $\delta J$ of the mass and angular momentum of the background within the Kerr family.\footnote{The well known C-metric and Kerr-NUT perturbations all support matter on some section of the symmetry axis of the background.} Since the reconstruction procedure in  \cite{Pound:2013faa,vandeMeent:2016pee} is performed separately in a region ``outside''  the particle orbit (i.e. a region of spacetime reaching from infinity to a timelike hypersurface $\mathcal{S}$ containing the particle worldline) and a region ``inside'' the particle orbit (i.e. a region of spacetime reaching from $\mathcal{S}$ to the horizon of the background metric.), this ambiguity occurs separately in each region leading two four unknown amplitudes $\delta M^\pm$ and $\delta J^\pm$, which need to be fixed using additional input. The problem of finding $\delta M^\pm$ and $\delta J^\pm$ is known as the \emph{completion problem}.

The values of  $\delta M^\pm$ and $\delta J^\pm$ affect the EMRI dynamics in two separate ways. First, they enter into the calculation of the conservative part of the first order self-force. Second, they will provide second order corrections to the flux of energy and angular momentum from the particle. In the two-timescale expansion of the dynamics \cite{Hinderer:2008dm}, both lead to an $\bigO(1)$ correction to the phase of emitted gravitational waveform over the inspiral timescale.

In \cite{Merlin:2016boc} the missing mass and angular momentum perturbations were recovered for particles on bound \emph{equatorial} orbits by requiring continuity of some gauge invariant fields constructed from the metric perturbation. The final result of the lengthy and involved calculation is remarkably simple: ``outside'' the particle's orbit the mass and angular momentum are given by the energy and angular momentum of the orbit, while both vanish ``inside'' the orbit.  As noted in \cite{Merlin:2016boc}, this result implies that in this instance the metric perturbations produced by the CCK procedure do not contain any mass or angular momentum, as measured by the  Abbott-Deser charges \cite{Abbott:1981ff}.

The goal of this paper is to prove this is true in a general sense: the metric perturbation produced through the CCK procedure from any regular harmonic solution of the homogeneous Teukolsky equation always has vanishing Abbott-Deser charges. As an immediate corollary we find that the simple result of \cite{Merlin:2016boc} is true for general bound orbits in Kerr space time, and independent of the variant (ingoing or outgoing) of the radiation gauge used. However, the lemma that is the main result of this paper is of more general applicability. In particular, it is not linked to to context of reconstructing the metric from vacuum perturbations of $\psi_0$ and $\psi_4$, but is applicable to other applications where the CCK procedure is used simply to generate vacuum metric perturbations from solutions of the Teukolsky equation.

The plan for this paper is as follows. In section \ref{sec:prelim} we introduce some of the preliminaries needed for our results. In particular, following \cite{Dolan:2012jg,Merlin:2016boc} we introduce the Abbott-Deser charges which will be used to measure the mass and angular momentum content of a perturbation. Section \ref{sec:lemma} introduces our main technical result in the form of a lemma and continue with its proof. The main consequences of this result are discussed in section \ref{sec:cor}. We conclude with a discussion of our results in section \ref{sec:discuss}.
\section{Preliminaries}\label{sec:prelim}

\subsection{The Kerr metric}
In this paper we study perturbations of the Kerr metric, which in Boyer-Lindquist coordinates is given by
\begin{equation}
\begin{split}
\label{eq:kerr}
g^\mathrm{Kerr}_{\mu\nu} \id{x^\mu}\id{x^\nu}=
&-\bh{1 - \frac{2r}{\Sigma}}\id{t}^2 
+ \frac{\Sigma}{\Delta} \id{r}^2
+ \frac{\Sigma}{1-z^2} \id{z}^2
\\
&+ \frac{1-z^2}{\Sigma} \bh{2a^2 M r (1-z^2)+(a^2+r^2)\Sigma}\id\phi^2
- \frac{4Mar(1-z^2)}{\Sigma}\id{t}\id\phi,
\end{split}
\end{equation}
where $z:=\cos\theta$, $\Delta:=r(r-2M)+a^2$, and $\Sigma:=r^2+a^2z^2$. Moreover, we use geometrized units ($G=c=1$). In (most of) the following we also set $M=1$ for convenience. All notations and conventions are compatible with the appendix of \cite{vandeMeent:2016pee}.

The CCK procedure (discussed below) further is formulated in the Newman-Penrose formalism. For that purpose we chose the following null tetrad (introduced by Kinnersley),
\begin{alignat}{3}
\tet{1}{\mu} &= l^\mu &&= \frac{1}{\Delta}(r^2+a^2,\Delta,0,a),\\
\tet{2}{\mu} &= n^\mu &&= \frac{1}{2\Sigma}(r^2+a^2,-\Delta,0,a),\\
\tet{3}{\mu} &= m^\mu &&= -\frac{\bar\rho \sqrt{1-z^2}}{\sqrt{2}}(\ii a,0,-1,\frac{\ii}{1-z^2}),\\
\tet{4}{\mu} &= \bar{m}^\mu &&= \frac{\rho \sqrt{1-z^2}}{\sqrt{2}}(\ii a,0,1,\frac{\ii}{1-z^2}),
\end{alignat}
where $l^\mu$ and $n^\mu$ are the in- and outgoing principle null vectors of the Kerr background.

\subsection{CCK metric perturbations}
The procedure developed by Chrzarnowski \cite{Chrzanowski:1975wv}, Cohen and Kegeles \cite{Cohen:1974cm,Kegeles:1979an}, and Wald \cite{Wald:1978vm}, also referred to as the ``CCK'' procedure, starts from solutions of the homogeneous Teukolsky equation with spin-weighted $s=\pm2$,
\begin{equation}
\begin{aligned}
\label{eq:kerrteuk} 
\cbB{-\Bh{\frac{(r^2+a^2)^2}{\Delta} -a^2(1-z^2)}\frac{\partial^2}{\partial t^2}
-\frac{4Mar}{\Delta}\frac{\partial^2}{\partial t\partial\phi}
-\Bh{\frac{a^2}{\Delta}-\frac{1}{1-z^2}}\frac{\partial^2}{\partial\phi^2}&
\\
+2s\Bh{\frac{M(r^2-a^2)}{\Delta}-r-iaz}\frac{\partial}{\partial t} 
+2s\Bh{\frac{a(r-M)}{\Delta}+\frac{i z}{1-z^2}}\frac{\partial}{\partial\phi}&
\\
+\Delta^{-s}\frac{\partial}{\partial r}\Bh{\Delta^{s+1}\frac{\partial}{\partial r}}
+\frac{\partial}{\partial z}\Bh{(1-z^2) \frac{\partial}{\partial z}}
-\frac{s(s+1)z^2-s}{1-z^2} &}{_s\Phi} =0.
\end{aligned}
\end{equation}
A vacuum metric perturbation is then constructed using,
\begin{equation}
h^{CCK}_{\mu\nu} := {_{s}\hat{\mathcal{H}}_{\mu\nu}}\circ{_s\Phi} +c.c.,
\end{equation}
where the $ {_{s}\hat{\mathcal{H}}_{\mu\nu}}$ are certain second order partial differential operators given by
\begin{equation}
\begin{split}
{_{+2}\hat{\mathcal{H}}_{\mu\nu}} :=
-\rho^{-4}\cbB{
&\tet{\mu}{1}\tet{\nu}{1}\bh{\bar\delta-3\alpha-\bar\beta+5\varpi}\bh{\bar\delta-4\alpha+\varpi}\\
&+\tet{\mu}{3}\tet{\nu}{3}\bh{\hat\Delta+5\mu-3\gamma+\bar\gamma}\bh{\hat\Delta+\mu-4\gamma}
\\
&-\tet{(\mu}{1}\tet{\nu)}{3}\Bh{
	\bh{\bar\delta-3\alpha+\bar\beta+5\pi+\bar\tau}\bh{\hat\Delta+\mu-4\gamma}\\
&\qquad\qquad	+\bh{\hat\Delta+5\mu-\bar\mu-3\gamma-\bar\gamma}\bh{\bar\delta-4\alpha+\pi}
	}
},
\end{split}
\end{equation}
and
\begin{equation}
\begin{split}
{_{-2}\hat{\mathcal{H}}_{\mu\nu}} := -&\cbB{
\tet{\mu}{2}\tet{\nu}{2}
\Bh{\delta+\bar\alpha +3\beta-\tau }\Bh{\delta+4\beta +3\tau} 
+\tet{\mu}{4}\tet{\nu}{4} \Bh{D-\varrho}\Bh{D +3\varrho}
\\
 & -\tet{(\mu}{2}\tet{\nu)}{4} \Bh{
 \hh{\delta-2\bar\alpha +2\beta-\tau}\hh{ D+3\varrho}
+\hh{D +\bar\varrho-\varrho}\hh{\delta +4\beta+3\tau}
}
}.
\end{split}
\end{equation}
Here $D,\hat\Delta,\delta,\bar\delta$ are the direction derivatives along the tetrad legs as used in the Newman-Penrose formalism. The other Greek symbols are the Newman-Penrose spin-coefficients (see e.g. the appendix of \cite{vandeMeent:2016pee}).

The produced metric perturbation satisfies the radiation gauge conditions. When $s=+2$ the metric satisfies the \emph{outgoing} radiation gauge (ORG) condition $n^\alpha h_{\alpha\beta}=0$, and when $s=-2$ the metric satisfies the \emph{ingoing} radiation gauge (IRG) condition $l^\alpha h_{\alpha\beta}=0$. In both cases the metric is traceless ${h^\alpha}_\alpha=0$.

Despite being a solution of the Teukolsky equation, the field ${_s\Phi}$ does \emph{not} coincide with the Weyl scalar $\psi_0$ or $\psi_4$ (of the appropriate spin-weight) that would be obtained from the corresponding CCK metric perturbation. The fields ${_s\Phi}$ are (in this context) known as Hertz potentials. A general inversion procedure for obtaining the Hertz potential corresponding to some $\psi_0$ or $\psi_4$ exists \cite{Ori:2002uv,vandeMeent:2015lxa}, but its details will not be discussed here.

\subsection{Abbott-Deser charges}
In this paper, following \cite{Dolan:2012jg} and \cite{Merlin:2016boc}, we measure the mass and angular momentum of a perturbation using conserved charges introduced by Abbott and Deser \cite{Abbott:1981ff} that can be defined quasilocally on any metric perturbation when the background has admits a Killing vector. Its advantage in the present context is that it applies to any vacuum perturbation of Kerr spacetime, whereas more conventional notions of mass and angular momentum cannot be applied directly. E.g. the Komar quantities require the existence of Killing vectors on the full (perturbed) metric (which we generically won't have), and the ADM quantities require asymptotic flatness (satisfied only by a subset of CCK perturbations. Nonetheless, in domains of common applicability (Kerr perturbations that are asymptotically flat or share the Killing symmetries of the background) the Abbott-Deser quantities coincide with the ADM and Komar ones \cite{Dolan:2012jg}.

For any metric perturbation $h_{\alpha\beta}$ on a background $g_{\alpha\beta}$ with a Killing vector $k^\mu$, Abbott and Deser \cite{Abbott:1981ff} introduced the anti-symmetric 2-form,
\begin{equation}\label{eq:ADform}
F_{\alpha\beta}[k^\mu]:= \frac{1}{8\pi}\hh{k^\lambda\CD{[\alpha}\bar{h}_{\beta]\lambda}+\bar{h}_{\lambda[\alpha}\CD{\beta]}k^\lambda-k_{[\alpha}\nabla^\lambda\bar{h}_{\beta]\lambda}
},
\end{equation}
where $\bar{h}_{\alpha\beta}:= h_{\alpha\beta}-(1/2)g_{\alpha\beta}{h_{\lambda}}^\lambda$ is the trace reversed metric perturbation. The key property of $F_{\alpha\beta}$ is that its divergence defines a conserved current,
\begin{equation}\label{eq:ADcurrent}
j^\alpha := -\CD{\lambda}F^{\lambda\alpha} = k_\lambda T^{\lambda\alpha},
\end{equation}
where $T_{\alpha\beta}$ is the energy-momentum tensor appearing as a source on the right hand side of the linearized Einstein equation for $h_{\alpha\beta}$. Consequently, $F_{\alpha\beta}$ is divergenceless for vacuum ( $T_{\alpha\beta}=0$) perturbations and can be used to define a topological charge,
\begin{equation}\label{eq:ADcharge}
\mathcal{Q}[h_{\mu\nu},k^\mu,\mathcal{S}] := \int_{\mathcal{S}}F^{\alpha\beta}[k^\mu]\id{\mathcal{S}_{\alpha\beta}}
\end{equation}
for any closed 2-surface $\mathcal{S}$. It was shown by Dolan and Barack \cite{Dolan:2012jg} that the Abbott-Deser charge $\mathcal{Q}$ is in fact invariant under gauge transformations. Moreover, Eq. \eqref{eq:ADcurrent} implies that if the closed surface envelopes a region containing a non-zero matter distribution $T_{\alpha\beta}$, the Abbott-Deser charge is equal to the Noether charge of the matter corresponding to the Killing vector $k^\mu$ \cite{Dolan:2012jg}. This last property is essential for completion problem as it allows one to related the Abbott-Deser charges of the metric perturbation to the energy and angular momentum of the particle source.

\section{Main lemma}\label{sec:lemma}
Using the symmetries of the background Kerr spacetime, solutions 
of the Teukolsky equation can be decomposed
\begin{equation}\label{eq:harm}
{_s\Phi}(t,r,z,\phi) = \int\id\omega\sum_{m} {_s\Phi}_{m\omega}(r,z)e^{i(m\phi-\omega t)},
\end{equation}
where each of the individual harmonic modes ${_s\Phi}_{m\omega}(r,z)e^{i(m\phi-\omega t)}$ satisfies the homogeneous Teukolsky equation and can thus be used as a starting point of the CCK procedure. It is in terms of these modes that we formulate our main lemma.

\begin{lem}[Main lemma]
Let ${_s\Phi}_{m\omega}(r,z)e^{i(m\phi-\omega t)}$ be a smooth solution of the homogeneous $(s=\pm2)$-Teukolsky equation. Then all Abbott-Deser charges of the corresponding vacuum CCK metric, $h_{\alpha\beta}^{CCK}$, vanish.
\end{lem}

\begin{proof}
Since we are dealing with vacuum metric perturbations, the Abbott-Deser charges are topological invariants of closed 2-surfaces. Hence we are free to choose our surfaces to be spheres $\mathcal{S}_{tr}$ at constant $t$ and $r$. This means that the integrals for the Abbott-Deser charges can be written,
\begin{equation}\label{eq:fixedintegral}
\mathcal{Q}[h_{\mu\nu},k^\mu,\mathcal{S}_{tr}]=\int\limits_{0}^{2\pi}\int\limits_{-1}^{1}\mathcal{F}[k^\mu] \id{z}\id{\phi},
\end{equation}
with
\begin{equation}
\mathcal{F}[k^\mu]:= \Sigma F^{rt}.
\end{equation}
We proceed by distinguishing three separate cases.
\paragraph{Non-stationary modes}
We first consider modes with $\omega\neq 0$. Since the CCK operators are partial differential operators that do not depend on $t$ explicitly, the metric (before taking the real part) will be proportional to $\exp(i\omega t)$. Since the construction of the AD-charge $\mathcal{Q}$ for $\mathcal{S}_{tr}$ also does not involve $t$ explicitly, this implies that $\mathcal{Q}$ is proportional $\exp(i\omega t)$. Since $\mathcal{Q}$ has to be independent of $t$, this can only be true if $\mathcal{Q}=0$. 

\paragraph{Non-axisymmetric modes}
For the second case we consider modes with $m\neq 0$. Since the CCK operators also do not depend on $\phi$ explicitly this means that $\mathcal{F}\propto \exp(i m \phi)$. Consequently, the integral in \eqref{eq:fixedintegral} vanishes.

\paragraph{Stationary and axisymmetric (SAS) modes}
We are left with the case $\omega=m=0$. As discussed in e.g. \cite{vandeMeent:2015lxa}, any stationary axisymmetric (SAS) solution can be written as a sum
\begin{equation}
{_s\Phi}_{SAS}(r,z) ={_s\Phi}_{SAS}^{+}(r,z) + {_s\Phi}_{SAS}^{-}(r,z),
\end{equation}
where ${_s\Phi}_{SAS}^{-}(r,z)$ satisfies regular boundary conditions at the horizon \cite{Barack:1999st},
\begin{align}
{_{-2}\Phi_{SAS}^{-}}(r,z) &= {_{-2}\Phi_0^{-}}(z)\Delta^2+\bigO(\Delta^3) &&\text{as }r\to r_{+}\\
{_{+2}\Phi^{-}_{SAS}}(r,z) &= {_{+2}\Phi_0^{-}}(z)+\bigO(\Delta) &&\text{as }r\to r_{+}
\end{align}
and ${_s\Phi}_{SAS}^{+}(r,z)$ satisfies regular boundary conditions at infinity,
\begin{align}
{_{-2}\Phi_{SAS}^{+}}(r,z) &= \frac{{_{-2}\Phi_0^{+}}(z)}{r^{-1}}+\bigO(r^{-2}) &&\text{as }r\to\infty\\
{_{+2}\Phi^{+}_{SAS}}(r,z) &= \frac{{_{+2}\Phi_0^{+}}(z)}{r^{-5}}+\bigO(r^{-6}) &&\text{as }r\to\infty.
\end{align}
We now calculate the Abbott-Deser charges of each component separately, starting with ${_s\Phi}_{SAS}^{-}(r,z)$. We first write a general Killing vector on Kerr spacetime as,
\begin{equation}\label{eq:killing}
k^\mu := (x,0,0,y),
\end{equation}
by explicit calculation and expanding near the horizon we find that $\mathcal{F}$ is given by
\begin{equation}
\mathcal{F}[k^\mu] \propto
 (2ay-r_{-} x)\Bh{\d{G(z){_{s}\Phi_0^{-}}(z)}{z}+\bigO(\Delta)}
\end{equation}
with
\begin{equation}
G(z) =\frac{-4(1-a^2)r_{+}z(1-z^2)}{\hh{2z^2 + r_{+}(1-z^2)}^2},
\end{equation}
where $r_{\pm}$ are the outer and inner horizon radii of the background Kerr spacetime, and the spin $s$ only affects to proportionality factor.

Consequently, assuming that the mode is smooth and therefore finite at the poles $z=\pm 1$ we conclude that
\begin{equation}
\mathcal{Q}=\int_{S^2} F^{\alpha\beta}[k^\mu]\id{\Sigma_{\alpha\beta}} = \bigO(\Delta).
\end{equation}
Since $\mathcal{Q}$ vanishes at the horizon it must vanish everywhere.

The procedure is similar for  ${_s\Phi}_{SAS}^{+}(r,z)$, where explicit calculation finds that
\begin{equation}
\mathcal{F}\propto r^{-1}
\end{equation}
near infinity, and consequently $\mathcal{Q}$ vanishes (as was already noted in \cite{Merlin:2016boc}).

We thus find the Abbott-Deser charges must vanish for any CCK metric perturbation constructed from a smooth Teukolsky mode.
\end{proof}\rmfamily

\section{Consequences}\label{sec:cor}
The main lemma implies that the Abbott-Deser charges must vanish for a broad class of Hertz potentials. This essentially includes all solutions of the Teukolsky equation for which the Fourier transform in $t$ and $\phi$ exists, and that are smooth enough such that any Fourier sums/integrals can be exchanged with the integrals of the Abbott-Deser flux.

However, the proof certainly does not cover all possible Hertz potentials. The form of the proof further suggests that the best place to look for counterexamples would be in stationary axisymmetric (SAS) modes that are singular on the symmetry axes $z=\pm1$. In the case $a=0$, Keidl et al. \cite{Keidl:2006wk} identified a number of such type-D solutions, some of which were identified as ``mass'' or ''angular momentum''  perturbations of the background.

In \ref{app:SASkernel} we repeat the calculation of Keidl et al. \cite{Keidl:2006wk} for general $a\neq 0$. Like \cite{Keidl:2006wk} we find an eight dimensional family of solutions. All solutions in this family have non-vanishing Abbott-Deser integrals. However, on closer examination this is due to the solutions being sourced by a non-vanishing energy-momentum distribution supported on the symmetry axis. On re-examination the same turns out to be true of the ``mass perturbations'' found in \cite{Keidl:2006wk}. The ``angular momentum perturbation'' from  \cite{Keidl:2006wk} however turns out to be a proper vacuum perturbation. It is given by
\begin{equation}
{_{-2}\Phi_{\delta J}} = i\frac{z(z^2-3)}{1-z^2}\delta J.
\end{equation}
As we guessed it is a Hertz potential that is singular at the poles $z=\pm1$. With the Killing vector as in \eqref{eq:killing}, the Abbott-Deser charge is
\begin{equation}
\mathcal{Q} =  y \delta J.
\end{equation}
This provides an effective counter example to the tempting conjecture that all metric perturbations constructed using the CCK formalism have vanishing Abbott-Deser charges; at least when $a=0$. For the spinning case the conjecture is in principle open, since there is a still larger class of potential Hertz potentials that we have not covered. (For example functions with a polynomial time dependence.)

Nonetheless, the class of Hertz potentials for which the Abbott-Deser charges vanish is wide enough to include any Hertz potential that would appear as the result of the inversion process utilized in \cite{vandeMeent:2015lxa,vandeMeent:2016pee} to obtain the radiation gauge metric perturbation generated by a point particle.

These papers use the `no string' formulation of the radiation gauge introduced by Pound et al. \cite{Pound:2013faa}, where the spacetime is divided in two halves by a hypersurface $\mathcal{S}$ containing the particle worldline and that separates the black hole horizon from (spatial) infinity. In each of these halves the metric perturbation is obtained using the CCK formalism from a Hertz potential that is regular at either the horizon or infinity. Consequently \cite{Wald:1973}, this procedure is ambiguous up to perturbations of the black hole mass and angular momentum (and possibly additional gauge terms). That is, in each half the reconstructed metric needs to be supplemented with a perturbation of the form
\begin{equation}
 \delta{M}^\pm \pd{g^\mathrm{Kerr}_{\mu\nu}}{M}+ \delta{J}^\pm \pd{g^\mathrm{Kerr}_{\mu\nu}}{J},
\end{equation}
where the partial derivative are to be taken with the mass $M$ and angular momentum $J=Ma$ held fixed, and the $\pm$ indicate the regions `outside' and `inside' $\mathcal{S}$ (and the particle's orbit). Finding the values of  $\delta{M}^\pm$ and  $\delta{J}^\pm$ is sometimes known as the completion problem.

In \cite{Merlin:2016boc}, a lengthy calculation involving the matching of gauge invariant fields from both sides of $\mathcal{S}$ was used to prove that for any bound orbit restricted to the equatorial plane of the black hole,
\begin{equation}\label{eq:completion}
\begin{aligned}
	\delta{M}^{+} &= E,	& 	\delta{M}^{-} &= 0,\\
	\delta{J}^{+} &= L,	& 	\delta{J}^{-} &= 0,
\end{aligned}
\end{equation} 
where $E$ and $L$ are the energy and (component parallel to total angular momentum of) the angular momentum of the particle. In section VI of \cite{Merlin:2016boc} it was further observed that since the Abbott-Deser mass of  $\pd{g^\mathrm{Kerr}_{\mu\nu}}{M}$ is 1 and its Abbott-Deser angular momentum vanishes and vice versa for $\pd{g^\mathrm{Kerr}_{\mu\nu}}{J}$, their result implied that the ``reconstructed part'' of the metric perturbation (in that particular scenario) had zero Abbott-Deser mass and angular momentum and conversely that proving so would be sufficient to determine the completion.

Since the Hertz potential in the `inside' and 'outside' regions (by construction) is always decomposable in smooth harmonics as in \eqref{eq:harm}, our main lemma implies that the reconstructed part of the metric perturbation always has vanishing Abbott-Deser mass and angular momentum, for \emph{any} bound orbit and regardless of the spin of the Hertz potential.

We thus obtain the following corollary to our main lemma, generalizing the main result of \cite{Merlin:2016boc}.
\begin{cor}[Completion amplitudes]
For a particle on any bound orbit around a Kerr black hole the `no string' radiation gauge metric perturbation, the completion amplitudes are given by Eq. \eqref{eq:completion}.
\end{cor}

It seems likely that this result also extends to unbound orbits and plunging trajectories, although some care is needed in examining the convergence of the mode-sum at infinity and/or the black hole horizon due to the presence of a distributional point source.

\section{Discussion}\label{sec:discuss}
We have established that the Abbott-Deser charges of any vacuum metric perturbation of Kerr spacetime generated from a regular harmonic mode of the Hertz potential using the CCK formalism, vanish. We thereby (partially) answer a long standing question regarding the mass and angular momentum content of such perturbations. In particular, this allows one to completely recover a vacuum metric perturbation (up to a regular gauge transformation) from the corresponding perturbation of the Weyl scalar $\psi_0$ or $\psi_4$, if one knows the Abbott-Deser mass and angular momentum.

As a corollary we find the mass and angular momentum perturbations needed to complete the `no string' radiation gauge metric perturbation generated by a point particle on any bound orbit around a Kerr black hole. We thereby generalize the previous result of \cite{Merlin:2016boc}, where the same result was obtained for the limited case of equatorial orbits using a much more elaborate calculation. The method set out in \cite{Merlin:2016boc}  nonetheless has value. Besides providing an independent verification of our present result, the method of \cite{Merlin:2016boc} can be extended to help smoothen the gauge modes in the no string radiation gauge \cite{gaugecompletion}, as is needed for some gravitation self-force calculations such as the self-force correction to the periapsis shift \cite{vandeMeent:2016hel}.

The extension of the results of \cite{Merlin:2016boc} to general (inclined) orbits is a key step towards the calculation of the gravitational self-force on such orbits, and thereby the study of the evolution of extreme mass-ratio inspirals; binary black hole systems consisting of a (super)massive black hole orbited by a stellar mass compact object. These so-called EMRIs form a key source of gravitational waves for the proposed space-based gravitational wave observatory LISA \cite{LISAproposal}. In particular, self-force on inclined orbits will be a key ingredient in studying the effect of orbital resonances \cite{Flanagan:2010cd,Flanagan:2012kg,Isoyama:2013yor,vandeMeent:2013sza,vandeMeent:2014raa}.
\section*{Acknowledgements}
The author would like to thank Leor Barack and Adam Pound for many productive discussions on this subject.
The author was supported by the European Research Council under the European Union's Seventh Framework Programme (FP7/2007-2013) ERC grant agreement no. 304978.

\appendix
\section{The SAS kernel}\label{app:SASkernel}
The Weyl scalars $\psi_0$ and $\psi_4$ are constructed from a metric perturbation $h_{\mu\nu}$ by certain second order differential operators. Composing these with the CCK metric construction operators produces a set of fourth order differential operators relating the spin $s=\pm2$ Hertz potentials to the Weyl scalars $\psi_0$ and $\psi_4$. Since $\psi_0$ and $\psi_4$ themselves satisfy the spin $s=\pm2$ Teukolsky equation, these operators map vacuum solutions of the Teukolsky equation into each other. In fact, these operators are (proportional to) the well-known Teukolsky-Starobinksy identities. It is the inversion of these operators that allows one to recover the Hertz potential from a physical Weyl scalar obtained through other means \cite{Ori:2002uv,Lousto:2002em}. Such an inversion will always be ambiguous up to an element of the kernel of these fourth order differential operators.

We here determine the stationary axisymmetric (SAS) component of this kernel. This was previously determined by Keidl et al. \cite{Keidl:2006wk} in the specific case of a Schwarzschild ($a=0$) background and a $s=-2$ (i.e. IRG) Hertz potential. We generalize their result to general Kerr ($a\neq 0$) backgrounds (and include both spins).

In these case of a Hertz potential ${_s\Phi}$ that is stationary and axisymmetric (i.e. independent of $t$ and $\phi$) the relevant fourth order differential equations relating the Hertz potential to $\psi_0$ and $\psi_4$ are
\begin{align}
32\bar\rho^{-4}\bar\psi_4 &= \Delta^2 \partial_r^4 \Delta^2 ({_{+2}\Phi}) &&= 4\bar\eth_{-1}\bar\eth_0\bar\eth_1\bar\eth_2 ({_{-2}\Phi})\label{eq:psi4inv}\\
8\bar\psi_0 &= \eth_{1}\eth_0\eth_{-1}\eth_{-2}({_{+2}\Phi})&&=8\partial_r^4 ({_{-2}\Phi}).\label{eq:psi0inv}
\end{align}
These equations involve either radial or angular derivatives but not both. The kernel is most easily found by setting $\psi_0$ and $\psi_4$ to zero and starting from the ``radial'' equations. Starting with the $s=-2$ case, the general solution is easily seen to be a third order polynomial in $r$ with arbitrary functions in $z$ as coefficients. However,  ${_{-2}\Phi}$ also needs to solve the $s=-2$ Teukolsky equation. Inserting the general solution, the left-hand side of Teukolsky equation is again (proportional to) a third order polynomial in $r$. Hence it can be solved order-by-order in $r$ yielding 4 linear second order ordinary differential equations in $z$ for the arbitrary functions. The most general stationary axisymmetric solution ${_{-2}\Phi_{Ker}}$ of both the homogeneous radial equation and the $s=-2$ Teukolsky equation is part of an 8 complex parameter family, which we find to be given by,
\begin{equation}\label{eq:homsol}
\begin{split}
{_{-2}\Phi_{Ker}}=\frac{1}{1-z^2}\cbB{
&c_1
+ (r-1)z c_2
+ z(z^2-3)c_3
+ \hh{r+(r-1)z^2} c_4
\\
&+ \hh{r^2-a^2}z c_5
+ r\hh{r(r-3)+3 a^2} c_6
+ \hh{r^2 +(r^2-a^2)z^2} c_7
\\
&+ rz\hh{\hh{r(r-3)+3a^2}z^2-3\hh{r(r+4)-4a^2}} c_8
}.
\end{split}
\end{equation}
It is then straightforward to check that this solution also satisfies the ``angular'' part of equation \eqref{eq:psi4inv}. Moreover, one readily checks that these solution coincide with the ones found in \cite{Keidl:2006wk} after setting $a=0$.

However, this is not the full story. The metric reconstruction procedure is only guaranteed to produce a vacuum solution of the Einstein equation on the coordinated patch that is being used. In our case this is the (modified) Boyer-Lindquist coordinate patch on the background Kerr spacetime, i.e. $r>r_{+}$ and $-1<z=\cos\theta<1$. In particular, there is the distinct possibility that these solutions are sourced by energy-momentum supported on the symmetry axes $z=\pm1$. Given that the Hertz potential \eqref{eq:homsol} is irregular on these axes, this seems more than a mere possibility.

If we calculate the Abbott-Deser flux \eqref{eq:ADcharge} through a cylinder enclosing a section of one of the symmetry axes, then, by construction, this is equal to the total mass and/or angular momentum within the cylinder as appearing in the energy-momentum tensor sourcing the metric perturbation. In order to get a proper vacuum perturbation we need to require, these charges to vanish in the limit that the radius of the cylinder is reduced to zero.

Imposing this condition for any section of the symmetry axis and any Killing vector, produces constraints on the coefficients $c_i$. In fact, when $a\neq0$ we find that the axes are only free of mass and angular momentum when all $c_i$ vanish. In other words, the SAS kernel for spinning Kerr spacetimes produces \emph{no} proper vacuum perturbations.

In the special case of a Schwarzschild $a=0$ background, the conditions become
\begin{equation}
\begin{split}
 \re c_2 =  \re c_4 = \re c_7 = \im c_2 &=\im c_4 = 0,\text{ and}\\
 \re(c_5+21 c_8)&=0
\end{split}
\end{equation}
or
 \begin{equation}
\begin{split}
 \re c_1 = \re c_2 = \re c_3 =  \re c_4 = \re c_5 =\re c_7 = \re c_8 &= 0,\text{ and}\\
 \im c_1 =\im c_2 =\im c_3 =\im c_4 &=0.
\end{split}
\end{equation}
In particular, we find that the ``mass perturbations'' identified in \cite{Keidl:2006wk} as the solutions with $\re c_4$ and $\re c_7$ non-zero, are in fact not vacuum perturbations, but have some energy momentum source associated with the symmetry axis. This misidentification is due to the gauge transformations used in the identification in \cite{Keidl:2006wk} being singular on the symmetry axes. We do however recover the angular momentum perturbation, in \cite{Keidl:2006wk} identified as the solution with $\im c_3$ non-zero, as a genuine vacuum perturbation.

The $s=+2$ case can be solved in similar fashion. The Hertz potential modes in the SAS kernel are simply,
\begin{equation}
{_{+2}\Phi_{Ker}} = \Delta^{-2}{_{-2}\Phi_{Ker}},
\end{equation}
as expected from some symmetries of stationary axisymmetric solutions of the Teukolsky equation (see e.g. \cite{vandeMeent:2015lxa}).

\raggedright
\section*{References}
\bibliography{../bib/journalshortnames,../bib/meent,../bib/commongsf,CCpaper}

\end{document}